\documentclass[conference]{IEEEtran}
\IEEEoverridecommandlockouts
\IEEEpubid{\makebox[\columnwidth]{978-1-6654-8045-1/22/\$31.00~\copyright2018 IEEE \hfill}
\hspace{\columnsep}\makebox[\columnwidth]{ }}
\usepackage{cite}
\usepackage{amsmath,amssymb,amsfonts}
\usepackage{amsthm}
\usepackage{algorithmic}
\usepackage{graphicx}
\usepackage{textcomp}
\usepackage{xcolor}
\usepackage{multirow}
\usepackage{caption}
\usepackage{subcaption}
\usepackage{url}

\def\BibTeX{{\rm B\kern-.05em{\sc i\kern-.025em b}\kern-.08em
    T\kern-.1667em\lower.7ex\hbox{E}\kern-.125emX}}
\newtheorem{definition}{Definition}
\newtheorem{theorem}{Theorem}
\newtheorem{property}{Property}

\begin{document}

\title{Shape-based Evaluation of Epidemic Forecasts\\
}

\author{\IEEEauthorblockN{Ajitesh Srivastava}
\IEEEauthorblockA{%\textit{Ming Hsieh Department of Electrical and Computer Engineering} \\
\textit{University of Southern California}\\
Los Angeles, United States of America \\
ajiteshs@usc.edu}
\and
\IEEEauthorblockN{Satwant Singh}
\IEEEauthorblockA{%\textit{Department of Computer Science} \\
\textit{University of Southern California}\\
Los Angeles, United States of America \\
satwants@usc.edu}
\and
\IEEEauthorblockN{Fiona Lee}
\IEEEauthorblockA{\textit{Stanford Online High School} \\
United States of America \\
leefiona@ohs.stanford.edu}
}

\maketitle
\IEEEpubidadjcol

\begin{abstract}
Infectious disease forecasting for ongoing epidemics has been traditionally performed, communicated, and evaluated as numerical targets -- 1, 2, 3, and 4 week ahead cases, deaths, and hospitalizations. While there is great value in predicting these numerical targets to assess the burden of the disease, we argue that there is also value in communicating the future trend (description of the shape) of the epidemic -- for instance, if the cases will remain flat or a surge is expected. To ensure what is being communicated is useful we need to be able to evaluate how well the predicted shape matches with the ground truth shape. Instead of treating this as a classification problem (one out of $n$ shapes), we define a transformation of the numerical forecasts into a ``shapelet''-space representation. In this representation, each dimension corresponds to the similarity of the shape with one of the shapes of interest (a shapelet). We prove that this representation satisfies the property that two shapes that one would consider similar are mapped close to each other, and vice versa. We demonstrate that our representation is able to reasonably capture the trends in COVID-19 cases and deaths time-series.
With this representation, we define an evaluation measure and a measure of agreement among multiple models. We also define the shapelet-space ensemble of multiple models as the mean of their shapelet-space representations. We show that this ensemble is able to accurately predict the shape of the future trend for COVID-19 cases and trends. We also show that the agreement between models can provide a good indicator of the reliability of the forecast.  
\end{abstract}


\section{Introduction}
Infectious disease forecasting and its communication is a crucial aspect of epidemic management. There have been various collaborative efforts for short-term forecasting during epidemics including Influenza seasons~\cite{centers2019flusight}, COVID-19~\cite{noauthor_covid-19_nodate,europe_forecast_hub,lab_covid-19_2020}, and Zika~\cite{cdc_epidemic_2016}, where multiple research groups independently generate forecasts that are combined to inform the public. Traditionally, the forecasts are communicated as time-series with prediction intervals. Evaluation and communication are closely related -- forecasts are rigorously evaluated to ensure that what is being communicated is meaningful. These evaluations are often performed by treating the problem as a regression problem -- comparing predicted values (e.g., incident hospitalizations and deaths) over a horizon (e.g., 1, 2,3, and 4 week ahead) to the observed values. 

\begin{figure}[!htb]
\centering
  \includegraphics[width=0.48\textwidth]{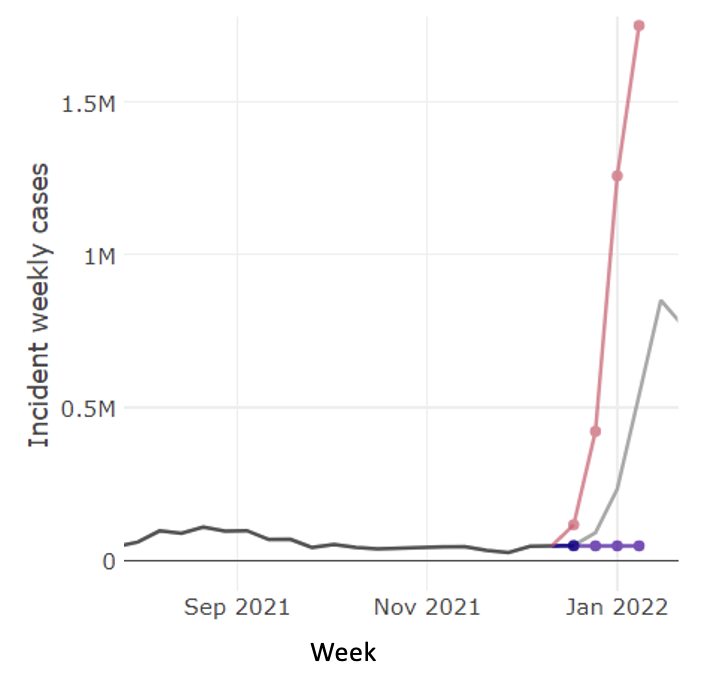}
  \caption{Value in Predicting Patterns: The purple forecast always predicts a flatline. The red forecast is able to predict the existence of a surge (grey line) but overpredicts the intensity. The red forecast is considered worse even though it has more useful information about the future.}
  \label{fig:motivation1}
\end{figure}
\begin{figure}[!htb]
\centering
  \includegraphics[width=0.48\textwidth]{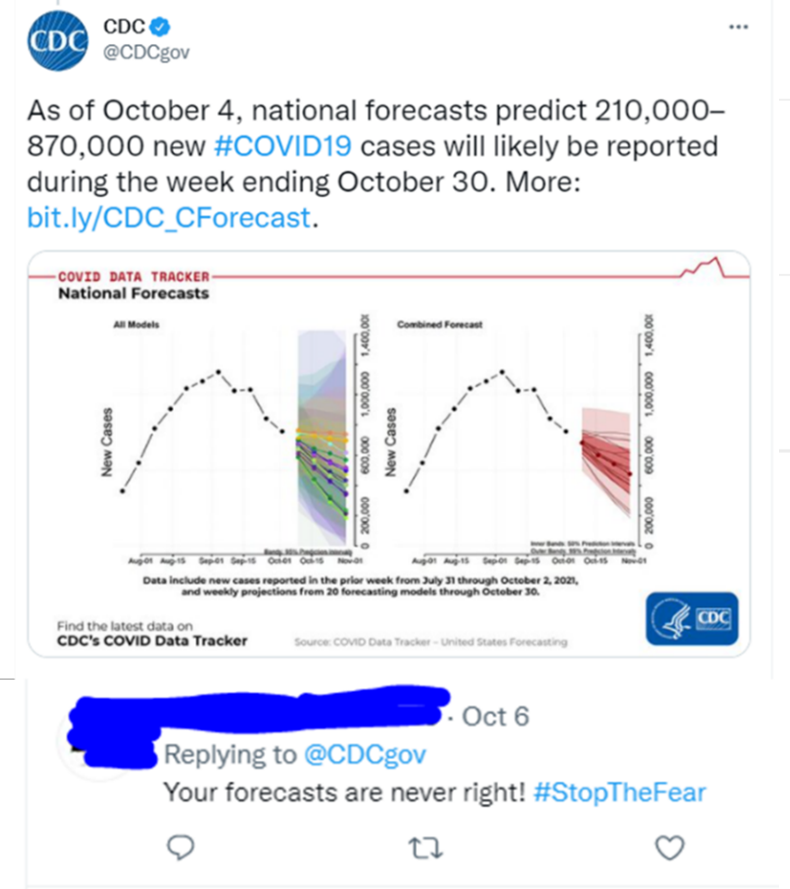}
  \caption{Communication of numerical forecasts can be overwhelming to the general public and may create distrust. While there is disagreement among the models regarding the exact numbers, there is agreement on the trend.}
  \label{fig:motivation2}
\end{figure}

While communicating time-series forecasts has value for assessing the burden of the disease, we argue that there is also value in communicating and evaluating a qualitative description of the future trend. This is due to the following reasons. First, the general public may be interested in knowing the future trends, such as ``is there going to be a surge'', and ``when will the cases turn around''.
Second, evaluation metrics for regression may not always convey the usefulness of a model. For instance, consider a model $M_1$ that is always able to predict if a surge is going to happen, but overestimates its severity (Figure~\ref{fig:motivation1}). Consider another model $M_2$ that always predicts a flat forecast -- each of the future weeks will have the same incidence as the recent week. The popular metric mean absolute error may assign a lower error to $M_2$ even though $M_1$ conveys more useful information than $M_2$.
Third, noise due to reporting delays affects the observed number of incident values that may deviate from the true incident values. This could penalizing forecasts that were even closer to the true incidents but slightly different from the noisy reported incidents. However, a qualitative description over a horizon (say ``increase'' or ``decrease'') is less sensitive to such noise.
Finally, there can be frequent disagreement between exact numbers in the time series predictions and uncertainty across models. Without an understanding of the model assumptions and approaches, the results can be overwhelming and may create an impression in the general public that the models are unreliable (see Figure~\ref{fig:motivation2}). 
\begin{figure}[!htbp]
\centering
  \includegraphics[width=0.48\textwidth]{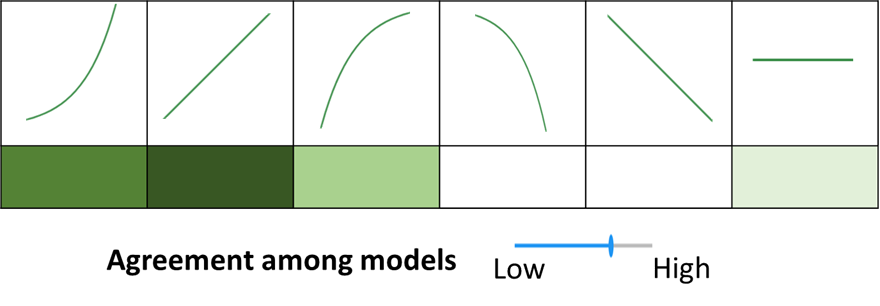}
  \caption{An example of communication of qualitative description of epidemic forecasts .}
  \label{fig:comm}
\end{figure}

To address the evaluation of forecasts based on a qualitative description, we propose a novel shapelet-based appraoch. In this approach, we define a set of $d$ shapes of interest called shapelets. The given time-series is then transformed into a $d$-dimensional space where each dimension encodes the similarity of the time-series to a shapelet. This shapelet-space representation allows us to compare two time series, which we can use to define evaluation measures as well as agreement across models. Since each dimension corresponds to a particular shapelet, the representation is interpretable and can be directly used for communication. As an example, Figure~\ref{fig:comm} shows a potential way to communicate the results to the public, where a ``heatmap'' represents the value associated with each dimension (shapelet) in the predicted shapelet-space.
Specifically, our contributions are as follows. (1) We propose a novel interpretable  time-series representation to evaluate short-term forecasts using the predicted shape and the ground truth shape. (2) We prove that our representation satisfies desirable properties that would map similar shapes to near points. (3) We demonstrate that the ensemble of shapes generated by transforming individual model predictions into our representation is able to predict the future shape more accurately than the shape obtained directly from the COVIDhub ensemble~\cite{cramer2022evaluation}. Specifically, we demonstrate that our ensemble of shapes can predict a shift in trend significantly better than the COVIDhub ensemble. (4) We also demonstrate that agreement among the models is a good indicator of reliability of the predictions. Our code to reproduce our results is available on Github\footnote{\url{https://github.com/Satwant-Singh-ADS/Shapelet_Methods}}.

\section{Related Work}
Various measures of errors have been used in infectious disease forecasting depending on the specific tasks as described below.

\textbf{Point forecasts} refer to the forecasting of one number for each ground truth of the future, such as the predicted number of new cases on a future date and declaring a future date when the peak is expected to happen. For such forecasts, mean absolute error (MAE) is defined for each prediction $\hat{y}_i$ corresponding to the ground truth $y_i$, $i \in \{1, \dots, n\}$,and a variation of mean percentage absolute error called symmetric MAPE or SMAPE~\cite{shcherbakov2013survey} is defined as:
\begin{equation}
    MAE = \sum_i \frac{|y_i - \hat{y}_i|}{n}\,.
\end{equation}
\begin{equation}
    SPAME = \frac{1}{n}\sum_i \frac{|y_i - \hat{y}_i|}{0.5|y_i + \hat{y}_i|}\,.
\end{equation}

Variations of both MAE and MAPE are widely used in time-series forecasting~\cite{chatfield2000time}. Particularly, MAE is the evaluation preferred by the CDC for point forecasts~\cite{cdc_epidemic_2016}.

\textbf{Probabilistic event forecasts} refer to the forecasting of a probability distribution for a ground truth event of the future, i.e., assigning a probability to each discrete possibility. 
For such forecasts, a log score~\cite{gneiting2007strictly} has been used by the CDC to evaluate real-time submissions of Flu forecasts~\cite{cdc_epidemic_2016}.
\begin{equation}
    LS= \frac{1}{n}\sum_i \max \{\ln{P(E_i)}, -10\}\,,
\end{equation}
where $P(E_i)$ is the probability assigned to the event $E_i$ that is observed in the ground truth. If the assigned probability is so low that $\ln{P(E_i)}$ is less than $-10$ or undefined, it is replaced by $-10$. This ensures that one significantly poor score does not affect the average. A higher score is preferred.

\textbf{Interval forecasts} refer to the reporting of a range with a confidence interval that suggests the likelihood of the true value falling in the range. One way to evaluate prediction intervals is ``coverage"~\cite{ray_ensemble_2020-1}, which measures the percentage of time the observed value falls within the provided interval for the given confidence (such as 95\% confidence interval). Other ways of evaluating interval forecasts and, more generally, quantile forecasts while penalizing long ranges also exist in the literature~\cite{gneiting2007strictly}. Currently, weighted interval score (WIS) is being widely used for the evaluation of COVID-19 forecasts and Influenza forecasts~\cite{bracher2021evaluating}.

For 1,2,3 and 4 week ahead forecasts MAE and WIS are the measures currently in use for FluSight and COVID-19 forecasting. Both measures are designed to numerically match the predictions and ground truth and are not designed to compare shapes.
To the best of our knowledge, no evaluation approach exists for infectious disease forecasting for comparing shapes of the predicted trends.
In time-series literature, ``shapelets'' have been used to capture various motifs that occur in time-series~\cite{ye2009time}. Here, we define a shapelet as a fundamental shape of interest and define a representation of any shape given by w-week ahead forecasts based on its similarity to our chosen shapelets.

\section{Methodology}
\begin{figure*}[!htbp]
\centering
  \includegraphics[width=0.8\textwidth]{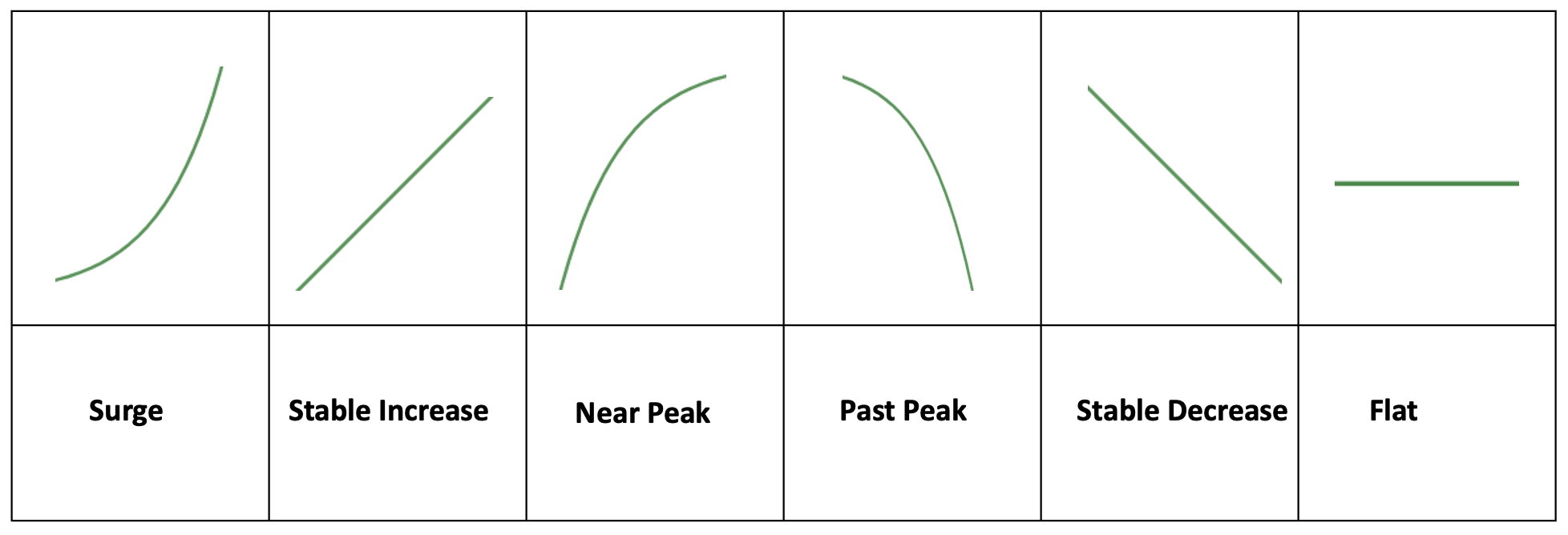}
  \caption{Shapelets of Interest}
  \label{fig:shapelets}
\end{figure*}

We define a shapelet as a vector that represents one of the shapes of interest in the trend of an epidemic. 
\begin{definition}[Shapelet]
 A shapelet $\mathbf{s} = [s_1, \dots, s_w] \in \mathbb{R}^w$ is a vector that represents a shape of interest.
\end{definition}
Figure~\ref{fig:shapelets} shows an example of a set of shapelets. One can treat the evaluation of forecasts as a classification problem. This would entail assigning the ground truth and the model prediction to one of the shapelets. 
However, treating this as a ``one out of $n$'' classification can incorrectly penalize forecasts. For example, classifying a ``surge'' as a ``stable increase'' should be penalized less than classifying it as a ``stable decrease'', as a ``surge'' is closer to the former but very different than the latter. Instead of classification, we assume that every shape has a representation determined by its similarity with each of the shapelets. Then two shapes are considered similar if their representations are close to each other. Prediction performance is measured based on the similarity of the predicted shape with the ground-truth shape based on these representations.

\begin{definition}[Shapelet-space Representation]
Given $d$ shapelets $\{\mathbf{s_1}, \dots \mathbf{s_d}\}$, we define a shapelet-space as a $d$-dimensional space where each point $P_x = (p_1, p_2, \dots, p_d)$ denotes a shape given by vector $\mathbf{x}$ and co-ordinate $p_i = sim(\mathbf{x}, p_i)$ for some measure of similarity.
\end{definition}

\subsection{Desired Properties}
% Two vectors have identical representation, i.e., $f(\mathbf{x}) = f(\mathbf{y})$, 
% (1.1) None of $x$ and $y$ are ``almost flat'', and (1.2) x can be obtained by scaling and translating. Or (2) both x and y are ``almost flat''.    
Let us consider the following property which may be desirable. Later, we will show that an additional property is needed for a sensible representation.
\begin{property}\label{prop:1}
Two vectors have similar representation, i.e., $|f(\mathbf{x}) - f(\mathbf{y})| \leq \delta$, for some small $\delta > 0$, if and only if $\mathbf{x}$ can be approximated by scaling and translating $y$. A weaker form of the property is that $f(\mathbf{x}) = f(\mathbf{y})$ if and only if $\mathbf{x}$ can be obtained by scaling and translating $y$.
\end{property}
We will show that choosing Pearson Correlation as the similarity measure $sim$ satisfies this property.
Consider five of the six shapelets shown in Figure~\ref{fig:shapelets}, described by the following for a 4-week ahead forecasting task:
\begin{itemize}
    \item Surge: $[1, 2, 4, 8]$
    \item Stable Increase: $[1, 2, 3, 4]$
    \item Near Peak: $[-1, -0.5, -0.25, -0.125]$
    \item Past Peak: $[-1, -2, -4, -8]$
    \item Stable Decrease: $[4, 3, 2, 1]$
%    \item Flat: $[0, 1, 0, 1]$
\end{itemize}

Consider a matrix $B$ such that each row corresponds to a shapelet above after a scaled standard normalization $\mathbf{s'} = \frac{\mathbf{s - \mu_s}}{\sqrt{w}\sigma_s}$, where $w$ is the length of each shapelet (here, $w = 4$). Here, $\mu_s$ is the mean of the entries in $\mathbf{s}$ and $\sigma_s$ is the standard deviation. Consider a vector $\mathbf{z}\in \mathbb{R}^w$ and its scaled standard normal form $\mathbf{z'}$ as defined before. Then, for shapelet $\mathbf{s}$, the Person Correlation between $\mathbf{z}$ and $\mathbf{s}$ is given by the dot product of $\mathbf{z'}$ and $\mathbf{s'}$. Further, $B\mathbf{z'}$ provides the representation of $z$ in the shapelet-space determined by the shapelets above.
\begin{theorem}\label{thm:prop1}
Given the shapelet-space as described above, Pearson Correlation satisfies the weaker form of Property~\ref{prop:1}. More precisely, two vectors are mapped to the same point if and only if one is obtained by translation and scaling of the other.
\end{theorem}
\begin{proof}
 We will first show that two vectors are mapped to the same point if and only if their standard normal forms are identical. Then, we will show that an identical standard normal form implies that one is obtained by translation and scaling of the other.
 
Consider vectors $\mathbf{x}, \mathbf{y} \in \mathbb{R}^{w}$, such that their standard normal forms $\mathbf{x'} = \mathbf{y'}$. Then $B\mathbf{x'} = B\mathbf{y'}$, and so they are mapped to the same point in the shapelet-space.

Now assume that for some $\mathbf{x'}$ and $\mathbf{y'}$, $B\mathbf{x'} = B\mathbf{y'}$. Then, $B(\mathbf{x'} - \mathbf{y'}) = \mathbf{0}$. Therefore, $\mathbf{x'} - \mathbf{y'}$ belongs to the null space of $B$. It can be shown that the null space of $B$ is the span of the vector $\mathbf{v} = [0.5, 0.5, 0.5, 0.5]$. So, for some $\lambda \in \mathbb{R}$, 
\begin{equation}\label{eqn:null}
    \mathbf{x'} - \mathbf{y'} = \lambda \mathbf{v}\,.
\end{equation}
Recall that $\mathbf{x'}$ and $\mathbf{y'}$ have zero mean. Taking average of each vector in Equation~\ref{eqn:null}, we get:
\begin{equation}
    0 - 0 = 0.5\lambda \implies \lambda = 0\,.
\end{equation}
Using the value of $\lambda$ in Equation~\ref{eqn:null}, we get
\begin{equation}
    \mathbf{x'} - \mathbf{y'} = \mathbf{0} \implies \mathbf{x'} = \mathbf{y'}\,,
\end{equation}

Two vectors $x$ and $y$ will be mapped to the same point in shapelet-space, if and only if they have the same standard normal form, i.e., 
\begin{equation}
    \frac{\mathbf{x} - \mu_x}{\sqrt{w}\sigma_x} = \frac{\mathbf{y} - \mu_y}{\sqrt{w}\sigma_y} \implies \mathbf{x} = \sigma_x\frac{\mathbf{y'} - \mu_y}{\sigma_y} + \mu_x\,.
\end{equation}
Therefore, given a shape, one can construct a new shape that will map to the same point by choosing $\mu_x \in \mathbb{R}$ and $\sigma_x \in \mathbb{R}^{+}$ corresponding to a translation and a scaling operation, respectively.
\end{proof}
\begin{theorem}\label{cor:prop1}
Given the shapelet-space described by the five shapelets above, two vectors are mapped close to each other if only if one vector is approximately obtained by translating and scaling the other.
\end{theorem}
\begin{proof}
Suppose one vector is approximately obtained by translating and scaling the other. Then, in our notation: $\mathbf{x'} - \mathbf{y'} = \mathbf{\epsilon}$ for some small $\| \epsilon\|$. Then, the vectors are mapped to nearby points as shown by 
\begin{align*}
    \| B\mathbf{x'} - B\mathbf{y'}\| = \|B\mathbf{\epsilon}\| \leq \|\mathbf{\epsilon}\|\,.
\end{align*}
Here the last inequality follows from applying Cauchy-Schwarz inequality to the dot product of each row of $B$ with the vector $\mathbf{\epsilon}$. 

%Similarly, note that $\| B\mathbf{z'} \| \leq \frac{1}{\sqrt{m}} \|z\|, \,\forall \mathbf{z} \in \mathbb{R}^m$. We will use this fact and the continuity of the linear operation to prove the converse.

Now, suppose $\mathbf{x'}$ and $\mathbf{y'}$ are mapped to nearby points, i.e., $B(\mathbf{x'} - \mathbf{y'}) = \mathbf{\delta}$, such that $\|\mathbf{\delta}\|$ is small. Given $\mathbf{\delta}$, we will attempt to find $\mathbf{z} = \mathbf{x'} - \mathbf{y'}$. The sum of elements of $\mathbf{z}$ must be $0$ as the sum of elements of both $\mathbf{x'}$ and $\mathbf{y'}$ are 0.
We note that the rank of our matrix $B$ is 3. We construct a matrix $C$ by picking 3 independent rows $r_1$, $r_2$ and $r_3$ from $B$ (these correspond to ``surge'', ``stable increase'', and ``near peak''). We append an additional row $[1, 1, 1, 1]$. Similarly, we construct another vector $\mathbf{\delta'}$ by picking the elements $r_1$, $r_2$, and $r_3$ from $\mathbf{\delta}$. Observe that $\mathbf{z}$ must satisfy 
\begin{equation*}
    C\mathbf{z} = \mathbf{\delta'}\,.
\end{equation*}
We find that the matrix $C$ is invertible. Therefore given $\mathbf{\delta}$, after constructing $\mathbf{\delta}$, there exists a unique $\mathbf{z}$ given by:
\begin{align*}
    \mathbf{z} &= C^{-1}\mathbf{\delta'} \\
    \implies \|\mathbf{x'} - \mathbf{y'}\| &\leq \|C^{-1}\|_F \|\mathbf{\delta'}\|\,
    \leq \|C^{-1}\|_F \|\mathbf{\delta}\|\,.
\end{align*}
The last inequality follows from the fact that $\mathbf{\delta'}$ is constructed by selecting elements from $\mathbf{\delta}$. We find that $\|C^{-1}\|_F = 28.71$. Therefore, when $\|\mathbf{\delta}\|$ is small, the distance between the standard normal forms of the vectors is also small. 
\end{proof}

\noindent\textbf{Consequence of choosing non-orthogonal shapelets.} Note that the vectors in the row of matrix $B$ are non-orthogonal. How does this affect a sense of distance? To understand this, note that $\|B\mathbf{x'} - B\mathbf{y'}  \|^2 = | (B\mathbf{x'})^T(B\mathbf{x'}) + (B\mathbf{y'})^T(B\mathbf{y'}) - 2(B\mathbf{x'})^T(B\mathbf{y'}) |$. Consider the singular value decomposition~\cite{klema1980singular} of $B = U \Gamma V^T$, where $U, V$ are unitary matrics (i.e., their columns are orthogonal unit vectors), and $\Gamma$ is a rectangular diagonal matrix. Then,
\begin{align*}
    &(B\mathbf{x'})^T(B\mathbf{y'}) = (U\Gamma V^T\mathbf{x'})^T (U\Gamma V^T\mathbf{y'})\\
    &= \mathbf{x'}^T V \Gamma^T (U^T U) \Gamma V^T \mathbf{y'}
    = \mathbf{x'}^T V \Gamma^T \Gamma V^T \mathbf{y'} \\
    &= (\Gamma V^T\mathbf{x'})^T (\Gamma V^T\mathbf{x'})
\end{align*}
Therefore, $\|B\mathbf{x'} - B\mathbf{y'}  \|^2 = \|\Gamma(V^T\mathbf{x'} - V^T\mathbf{y'})\|^2$, i.e.,  the distance in our transformation is equivalent to a weighted distance in a transformation obtained by projections on a set of orthogonal unit vectors.

\noindent\textbf{Consideration of the ``flat'' shapelet.} Note that we have not considered the ``flat'' shapelet in the above analysis ($\mathbf{s_i} = [0, 0, 0, 0]$. This is an important shapelet to consider as it indicates a lack of significant change in dynamics. However, unlike other shapelets, ``flat'' is scale dependent -- a line with a negligible non-zero slope that should be considered flat can have an arbitrarily high slope depending on the scaling. Intuitively, the a shape is to be considered similar to a non-flat shapelet only if there is a ``significant'' change within the $w$-element vector $\mathbf{x}$. Therefore, we modify Property~\ref{prop:1}.

\begin{property}\label{prop:2}
Two vectors have similar representation, if and only if (i) none of the vectors are ``almost flat'' and one can be approximately obtained by scaling and translating the other, or (ii) both vectors are ``almost flat''.  
\end{property}

To find the shapelet-space representation, we first identify how similar it is to what we could consider ``flat'', and then update the similarities of the shape with respect to other shapelets. Mathematically, in the given scale, for some constants $m_0, \beta \geq 0$, we define ``flatness'' as 
\begin{displaymath}
       \phi =
        \left\{\begin{array}{@{}cl}
                1, & \text{if } m \leq m_0,\\
                \exp(-\beta(m-m_0)),   & \text{if } m \geq m_0.
        \end{array}\right.
\end{displaymath}
Here $m$ is the average absolute slope of the vector $\mathbf{x}$ whose shapelet-space representation is desired, i.e., if $\mathbf{x} = (x_1, x_2, x_3, x_4)$, then $m = (|x_2 - x_1| + |x_3-x_2| + |x_4 - x_3|)/3$. The constant $m_0$ represents the threshold of the average slope below which the shape needs to be considered perfectly ``flat'', i.e., it is not to be considered as even partially similar to other shapelets. The constant $\beta$ represents how quickly above the threshold $m_0$, the ``flatness'' should reduce. As an example, suppose we have scaled the time-series so that a slope of 1 should have a small flatness $\varepsilon$. Also, suppose we have set $m_0 = 0$. Then $\beta = -\ln \varepsilon$.
Now, the co-ordinates of shapelet-space representation are defined as
\begin{displaymath}
       sim(\mathbf{x}, \mathbf{s_i}) =
        \left\{\begin{array}{@{}cl}
                2\phi-1, & \text{if } \mathbf{s_i} \text { is ``flat''},\\
                (1-\phi)corr(\mathbf{x}, \mathbf{s_i}),   & \text{otherwise}.
        \end{array}\right.
\end{displaymath}

The above formulation ensures that the representation is not significantly impacted by scaling when the shape is not close to ``flat''. 

\begin{theorem}\label{thm:prop2}
The shapelet-space representation described above satisfies Property~\ref{prop:2}
\end{theorem}
\begin{proof}
If both vectors $\mathbf{x}$ and $\mathbf{y}$ are ``almost'' flat, then, by definition, $\phi_x$ and $\phi_y$ are close to $1$. Then, for both vectors, the similarity with flat is close to 1, while with others the similarity is close to zero, due to the $(1-\phi)$ term. Therefore, both vectors have similar representations in shapelet-space.

Next, we will show that if none of the vectors $\mathbf{x}$ and $\mathbf{y}$ are ``almost flat'' (i.e., $\phi$ is small) and both have an approximately similar standard normal form, then they are mapped close to each other.
Consider two vectors $\mathbf{x}$ and $\mathbf{y}$ with corresponding average slopes $m_x$ and $m_y$. Suppose, $m_x > m_y > 1$. Then, by our choice of $\beta$, $\phi_x < \phi_y < \varepsilon$. Therefore, $\phi_y - \phi_x < \varepsilon$. Therefore, when shapelet $\mathbf{s_i}$ ``flat''
\[
       |sim(\mathbf{y}, \mathbf{s_i}) - sim(\mathbf{x}, \mathbf{s_i})| =
                2(\phi_y-\phi_x) \leq 2\varepsilon,
                \]
Also, since they have approximately similar standard normal form, i.e., $\mathbf{x'} - \mathbf{y'} = \mathbf{\epsilon}$ with small $\|\mathbf{\epsilon}\|$, for non flat shapelet $s_i$, $corr(\mathbf{x}, \mathbf{s_i}) - corr(\mathbf{y}, \mathbf{s_i}) = \mathbf{s'_i}^T(\mathbf{y'} + \mathbf{\epsilon}) - \mathbf{s'_i}^T \mathbf{y'}$ . Recall that $\mathbf{s'_i}$ is the $i^{th}$ row of matrix $B$. So, for non-flat shapelets,

\begin{align}
        |sim(\mathbf{y}, \mathbf{s_i}) &- sim(\mathbf{x}, \mathbf{s_i})|  \nonumber\\ 
        &=|(1-\phi_y)corr(\mathbf{y}, \mathbf{s_i}) - (1-\phi_x)corr(\mathbf{x}, \mathbf{s_i})| \nonumber\\
        &=|(1-\phi_y)\mathbf{s'_i}^T \mathbf{y'} - (1-\phi_x)\mathbf{s'_i}^T (\mathbf{y'} + \mathbf{\epsilon})| \nonumber\\
        &= |(1-\phi_y - 1 + \phi_x)\mathbf{s'_i}^T\mathbf{y_i} + (1-\phi_x)\mathbf{s'_i}^T\mathbf{\epsilon}|  \nonumber\\
        &\leq\varepsilon |\mathbf{s'_i}^T\mathbf{y_i}| + (1-\phi_x)\|\mathbf{\epsilon}\| \leq \varepsilon + \|\mathbf{\epsilon}\|\,.
\end{align}
Therefore, 
\begin{align*}
|f(\mathbf{x}) - f(\mathbf{y})| &\leq \sqrt{(2\varepsilon)^2 + (d-1)(\varepsilon + \|\epsilon\|)^2}\,.
\end{align*}

Now we will show the converse. Consider vectors $\mathbf{x}$ and $\mathbf{y}$ such that $f(\mathbf{x} - \mathbf{y}) = \mathbf{\delta}$, where $\|\mathbf{\delta}\|$ is small. Since each dimension of $\mathbf{\delta}$ is small, the dimension corresponding to the flat shapelet is also small, i.e., for some $\epsilon$ 
\begin{equation}
|2(\phi_x)-1 - 2(\phi_y) + 1| = \epsilon \implies |\phi_x - \phi_y| = \epsilon/2\,.    
\end{equation}
Without loss of generality, we assume $\phi_y = \phi_x + \epsilon/2$. For non-flat shapelets, we have
\begin{align*}
    |\delta_i| &= |sim(\mathbf{y}, \mathbf{s_i}) - sim(\mathbf{x}, \mathbf{s_i})|\\
    &= |(1-\phi_y)corr(\mathbf{y}, \mathbf{s_i}) - (1-\phi_x)corr(\mathbf{x}, \mathbf{s_i})|\\
    &= |(1-\phi_x - \epsilon/2)corr(\mathbf{y}, \mathbf{s_i}) - (1-\phi_x)corr(\mathbf{x}, \mathbf{s_i})| \\
    &= |(1-\phi_x)(corr(\mathbf{y}, \mathbf{s_i}) - corr(\mathbf{x}, \mathbf{s_i})) - \frac{\epsilon}{2}corr(\mathbf{y}, \mathbf{s_i})|\\
    &\leq (1-\phi_x)|\mathbf{s'_i}^T(\mathbf{y} - \mathbf{x})| + \frac{\epsilon}{2}
\end{align*}
 The above holds, if $\phi_x$ is close to 1. Since $\phi_y = \phi_x + \epsilon/2$, this means $\phi_y$ is also close to 1, i.e., both $\mathbf{x}$ and $\mathbf{y}$ are almost flat. If $\phi_x$ is not close to 1, then $|\mathbf{s'_i}^T(\mathbf{y} - \mathbf{x})|$ must be small for all $i$. From Theorem~\ref{cor:prop1}, this leads to $\mathbf{x'}$ and $\mathbf{y'}$ being approximately equal.
 \end{proof}
 
 \noindent\textbf{Choosing the set of shapelets. } We have performed the analysis with five specific shapelets, along with the flat shapelet. However, based on the above analysis, we can identify how to select a set of shapelets so that Property~\ref{prop:2} is satisfied. For the task of $1, 2, ..., w$ ahead forecast, the proof of Theorem~\ref{thm:prop2} relies on Theorem~\ref{cor:prop1} that requires the matrix $C$ to be invertible. Recall that $C$ is constructed by picking $w-1$ rows from $B$ corresponding to its rank, and then appending a row $[1, 1, \dots, 1]$ (i.e., a row of $w$ $1s$). Note that if the rank of matrix B were to be $w$, then we would construct $C$ simply by picking $w$ independent rows from $B$. The appending of a row of $1s$ is no longer required as $C$ would already be invertible. The appending of the row only serves as a constraint that is known to be satisfied.
 Therefore, other than ``flat'', one should first select at least $w-1$ shapelets that are linearly independent vectors. Other shapelets may be linear combinations of these vectors. While the inclusion of other shapelets seems redundant, we allow them so that the interpretations can be generalized for any shapelet -- a high positive value at dimension $i$ means the predicted shape is close to shapelet $i$.     
 
\subsection{Performance Evaluation}
To evaluate the performance of $1, 2, \dots, w$ week ahead predictions by a model $\mathbf{x_i}$, we obtain its shapelet-space representation $f(\mathbf{x_i})$ and compare it to the shapelet-space representation $f(\mathbf{g})$ ground truth vector $\mathbf{g}$. The comparison is done through cosine similarity. Specifically, we define shapelet-space score as
\begin{equation}\label{eqn:score}
    SS(\mathbf{x_i}, \mathbf{g}) = \frac{f(\mathbf{x_i})f(\mathbf{g}) }{\|f(\mathbf{x_i})\|\|f(\mathbf{g})\|}.
\end{equation}
We favor cosine similarity over Euclidean distance for performance measure as cosine similarity results in a bounded measure in $[-1, 1]$. To assess the agreement between predictions from $n$ available models, we take the average pairwise similarity across the representations of all the available models. Specifically, we define inter-model agreement as
\begin{equation}\label{eqn:ima}
    IMA = \frac{2}{n(n-1)}\sum_{i>j}\frac{ f(\mathbf{x_i})f(\mathbf{x_j}) }{\|f(\mathbf{x_i})\|\|f(\mathbf{x_j})\|}.
\end{equation}        
Being as the average of cosine similarities, the agreement also lies in the interval $[-1, 1]$. It should be noted that, with multiple models, while achieving an agreement of $1$ is possible, achieving an agreement of $-1$ may not be possible. This is due to the fact that it is impossible to have more than two vectors such that all pairs point to the opposite direction in space.

\subsection{Shapelet-space Ensemble}

The traditional ensemble technique for short-term numerical epidemic forecasting is to take the mean or median of all the submitted forecasts. Instead, we want to capture the ensemble of shapes. Therefore, we define the shapelet-space based ensemble (in short, Shapelet Ensemble) as the centroid of the points given by the shapes of individual model forecasts. Mathematically, Shapelet Ensemble of $n$ models with numerical forecasts vector $\mathbf{x_i}$ is given by $(\sum_i f(\mathbf{x_i}) )/n$.

\section{Experiments}
The objective of the experiments is to demonstrate the following. (i) the proposed methods generate meaningful shapelet-space representations; (ii) Shapelet Ensemble produces good predictions of future shapes; and (iii) the inter-model agreement is a good indicator of the reliability of the forecasted shape.
\subsection{Data}
We conducted experiments on the short-term cases and deaths forecast submissions to the COVID-19 forecast hub~\cite{lab_covid-19_2020}. Each week, the models submitted 1, 2, 3, and 4 week-ahead forecasts at the state-level as well as the national level. The point forecasts form vectors of size $w = 4$. They also submitted quantiles for these targets. Evaluation of the probabilistic forecasts is left to future work. The ground truth was obtained from John's Hopkins COVID-19 dataset~\cite{noauthor_2019_nodate}. We used a moving average smoothing on the incident ground truth with window size 3. This was performed to induce smooth shapes in the ground truth, which will then be compared with the model forecasts.

\subsection{Shapelet Representation Analysis}
To compute the shapelet-space representation, we defined the flatness by setting $m_0 = 0$ and $\beta = -\log(0.1)/\theta$. Here, for the given time-series of cases or deaths, $\theta$ is the maximum average of the absolute value of the slope over 4-week vectors observed until Jun 28, 2020, for the given time-series of cases or deaths. This results in flatness of a shape with average absolute slope $m$ to be $(0.1)^{m/\theta}$. The motivation behind this choice was that by June 28, 2020, the US states had already seen an uptick in cases and deaths. We used the highest rate of change up to this day as a base slope $\theta$ that should get a low flatness value. Due to this choice, for any future slope $m = \theta$, the flatness will be evaluated as 0.1.
One may select a different set of $m_0$ and $\beta$ to fit their definition of what should be considered ``flat''. Other reasonable choices of $m_0$ and $\beta$ did not affect the conclusions of the experiments.

We computed the shapelet-space representation of the ground truth and identified the dimension with the highest value. This is annotated in Figure~\ref{fig:CA_annot}. For instance, a label $D$ indicates that the shape of the next four weeks has the highest value in the dimension corresponding to ``Decreasing'' in the shapelet-space representation. 
Observe that the annotations seem reasonable.  There are certain parts where the annotations switch quickly, e.g., a sequence of ``N''(Near peak) and ``I'' (Increase) at the beginning of the time-series. However, the overall representations at these points are very similar to each other with small differences at ``N'' and ``I'' dimensions. To demonstrate this we define trend continuity (TC) -- a measure of how similar one shape is to the previous shape in the time-series. Mathematically, we find the cosine similarity between consecutive shapes:
\begin{equation}
    TC(t) = \frac{f(\mathbf{g}(t-1))^T f(\mathbf{g}(t)) }{\|f(\mathbf{g}(t-1))\|\|f(\mathbf{g}(t))\|}.
\end{equation}
Here, $\mathbf{g}(t)$ is the ground truth shape determined by the vector of four elements (cases or deaths) at time $t+1, \dots, t+4$. Observe that at the early part of the time-series, while the annotations change, the trend continuity remains close to 1. We also observe that the trend continuity falls as the time-series goes through changes in trends.

We also observe that the inter-model agreement often remains high (Figure~\ref{fig:CA_agreement}). Even though the forecasts may disagree in terms of numerical targets, they often agree on the future shape.
Comparing with the trend continuity in Figure~\ref{fig:CA_annot}, the agreement seems to fall during the times of shift in trends.
\begin{figure*}[!htbp]
\centering
  \includegraphics[width=0.8\textwidth,trim={0 0 0 1.1cm},clip]{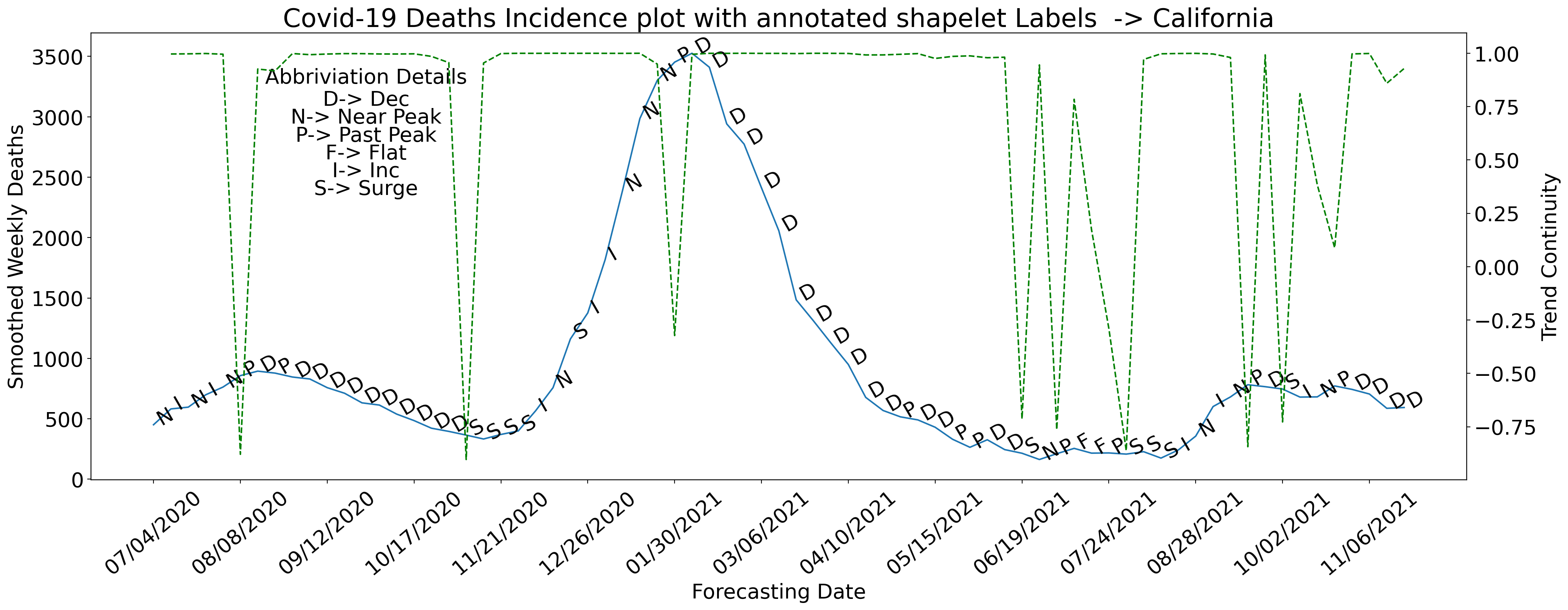}
  \caption{Shapelet-space dimension with highest value over time along with trend continuity (green curve) for COVID-19 weekly deaths in California.}
  \label{fig:CA_annot}
\end{figure*}

\begin{figure*}[!h]
\centering
  \includegraphics[width=0.8\textwidth,trim={0 0 0 1.1cm},clip]{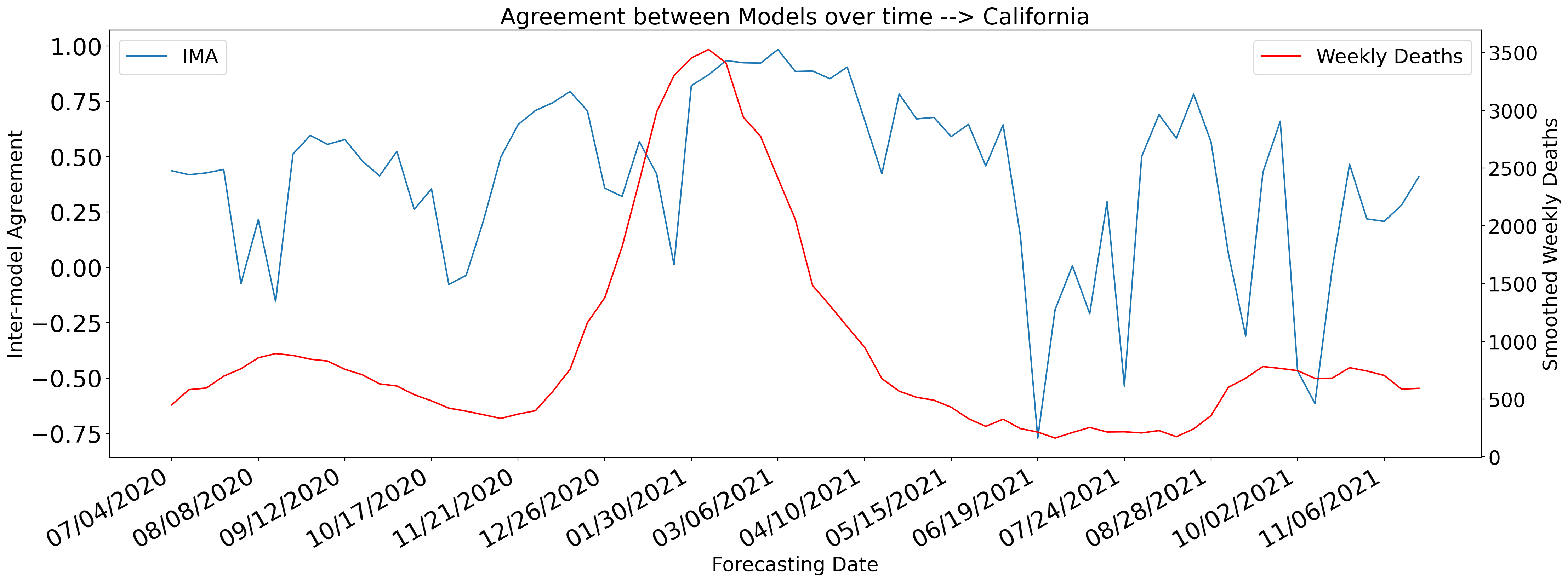}
  \caption{Measure of agreement between different models over time for death forecasting in California}
  \label{fig:CA_agreement}
\end{figure*}

\begin{figure*}[!h]
\centering
  \includegraphics[width=0.8\textwidth,trim={0 0 0 1.1cm},clip]{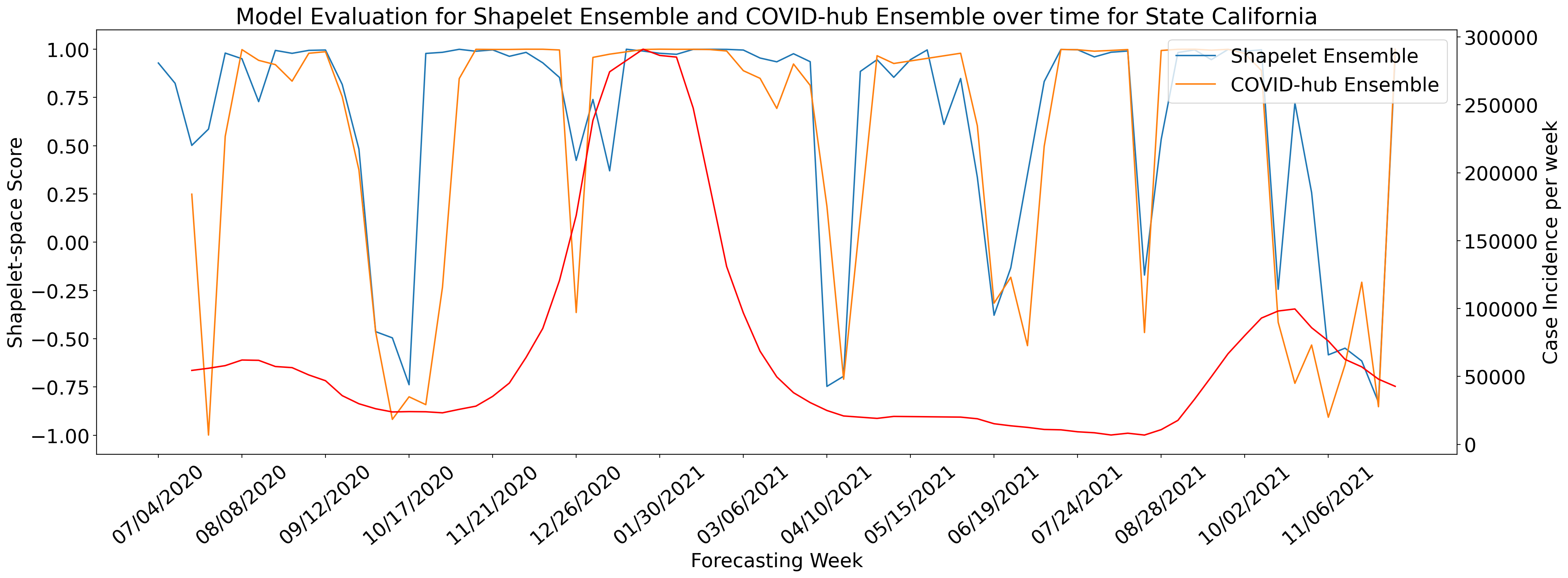}
  \caption{Performance of COVIDhub ensemble and Shapelet ensemble for forecasting shapes of cases in California.}
  \label{fig:CA_comparison_cases}
\end{figure*}

\begin{figure*}[!h]
\centering
  \includegraphics[width=0.8\textwidth,trim={0 0 0 1.1cm},clip]{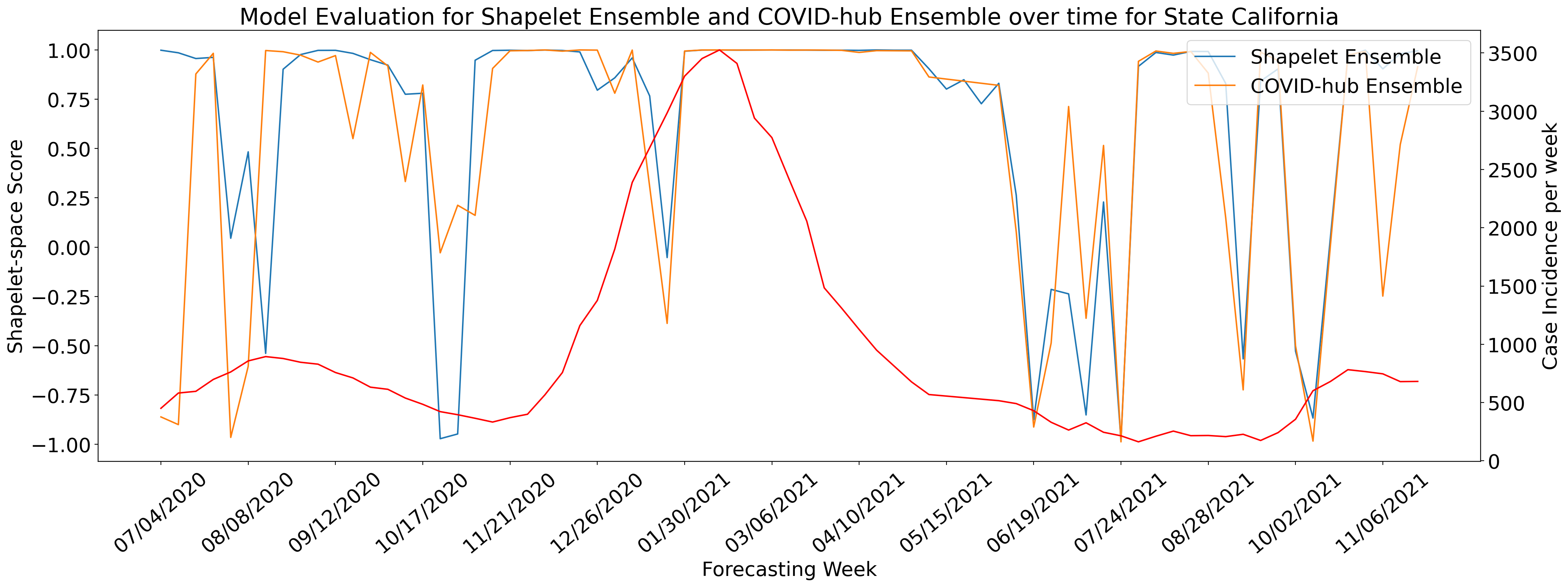}
  \caption{Performance of COVIDhub ensemble and Shapelet ensemble for forecasting shapes of deaths in California.}
  \label{fig:CA_comparison}
\end{figure*}

\begin{figure*}[!h]
\centering
  \includegraphics[width=0.8\textwidth,trim={0 0 0 1.1cm},clip]{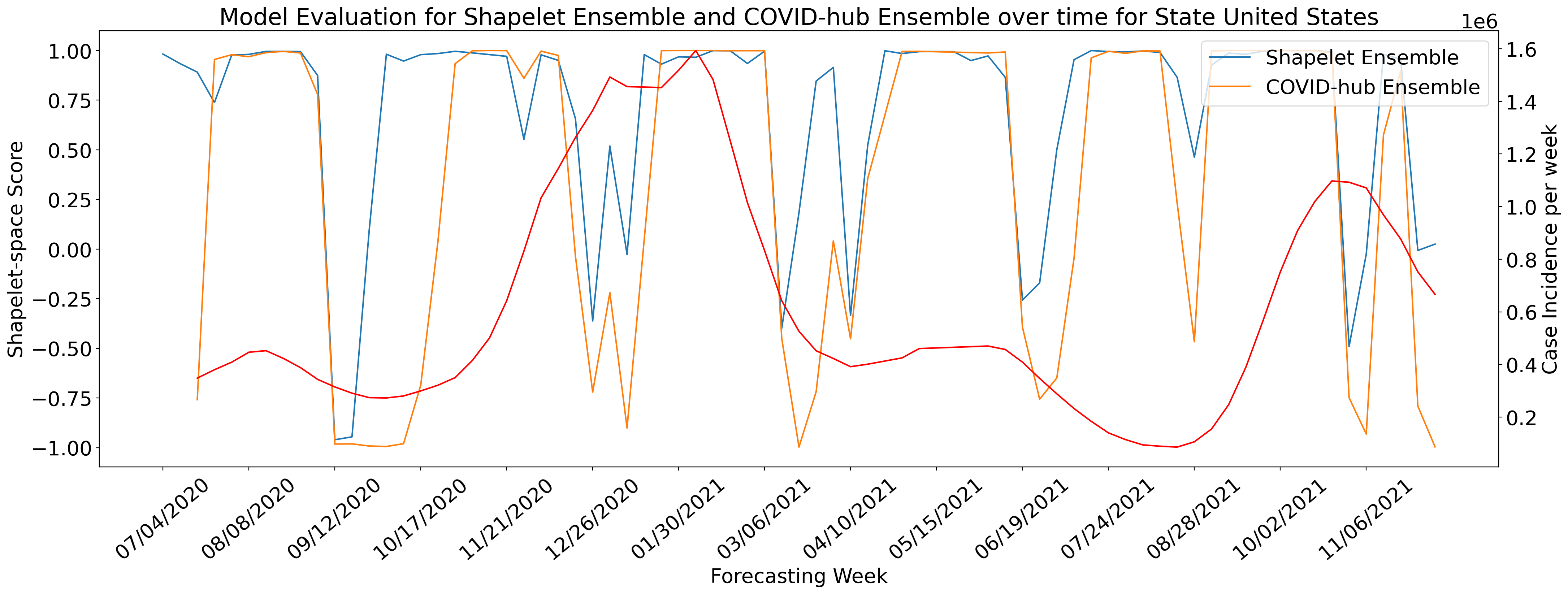}
  \caption{Performance of COVIDhub ensemble and Shapelet ensemble for forecasting shapes of cases at the national level.}
  \label{fig:US_comparison_cases}
\end{figure*}

\begin{figure*}[!h]
\centering
  \includegraphics[width=0.8\textwidth,trim={0 0 0 1.1cm},clip]{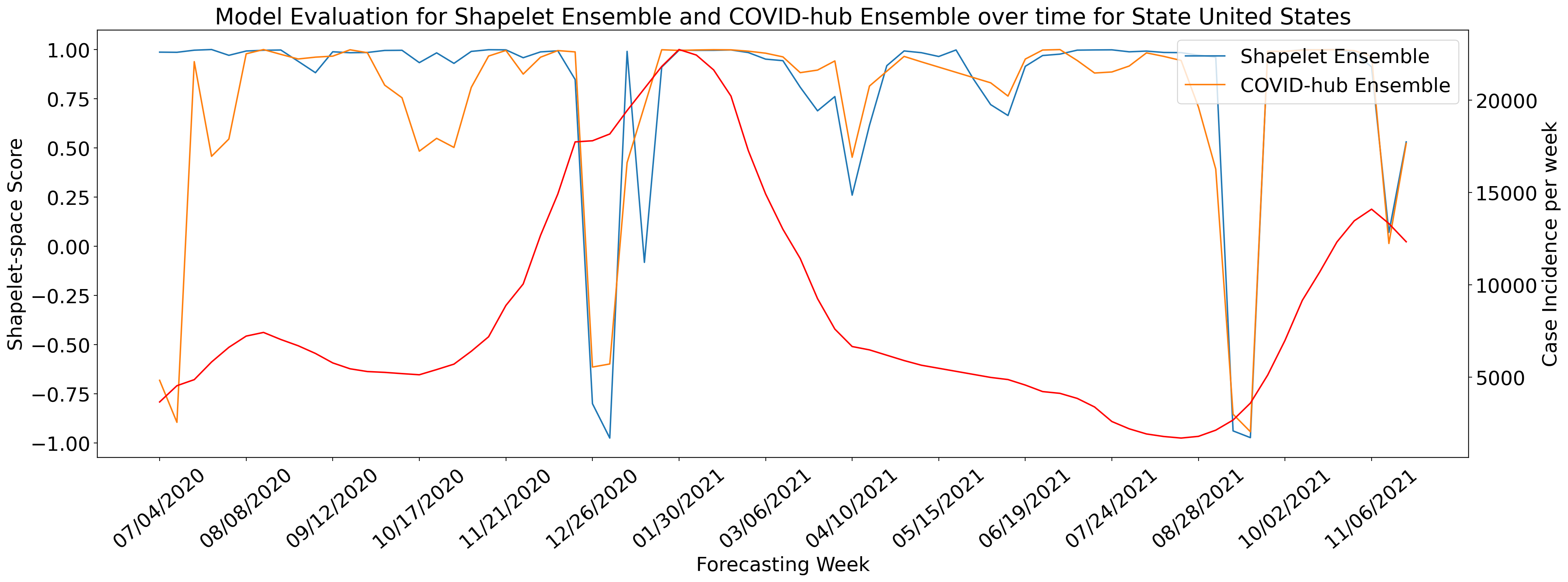}
  \caption{Performance of COVIDhub ensemble and Shapelet ensemble for forecasting shapes of deaths at the national level.}
  \label{fig:US_comparison}
\end{figure*}

\subsection{Performance Evaluation}
We computed the performance of all models based on shapelet-space scores across the US states for both weekly deaths and cases. We focus our evaluations on the two ensembles: (i) \textbf{COVIDhub Ensemble} - the shapelet-space representation of the COVID19hub ensemble which combines the numerical forecasts from all the submitted forecasts using mean or median; and (ii) \textbf{Shapelet Ensemble} - the mean of shapelet-space representation of all the submitted models. Figure~\ref{fig:CA_comparison_cases} shows the scores over time for national level case forecasts. Note that both models perform well for the majority of the time period, particularly for phases where the trend does not change. When there is a shift in trend, the performance of both models drop. However, the Shapelet Ensemble seems to catch the change earlier than COVIDhub Ensemble. A similar pattern is observed for other states and at the national level (Figure~\ref{fig:US_comparison_cases}). This suggests that the information on shifting trends is often present in the individual models but perhaps is lost when the mean/median of the numerical forecasts is taken. Figures~\ref{fig:CA_comparison} and~\ref{fig:US_comparison} show the scores for the two models for prediction of deaths. The earlier prediction of shifting trends by Shapelet Ensemble is not clear for deaths. This may be due to the fact that shifting trends in deaths follow observed shifting trends in cases (as cases cause deaths). Therefore, many models can predict the shape correctly, and so the COVIDhub Ensemble is able to capture the change.

Table~\ref{tab:results} shows the average performance of the models over all the weeks and all states. For both cases and deaths, Shapelet Ensemble performs significantly better, with the difference being higher for cases. To study how these models perform at different phases of the pandemic, we considered two types of trends: (i) Continued Trend - when trend continuity $TC(t) \geq 0$, and (ii) Changing Trend - when $TC(t) < 0$. Shapelet Ensemble has a much higher performance at times of Changing Trends compared to COVIDhub Ensemble. Particularly for cases, COVIDhub Ensemble has a poor performance close to 0 (recall that the range is $[-1, 1]$). 

\begin{table}[!ht]
\caption{Comparison of performance of Shapelet Ensemble vs COVIDhub Ensemble over all states (i) during all time points, (ii) during phases of changing trends, and (iii) during phases of continued trend}
\label{tab:results}
\begin{tabular}{llccc}
\hline
\multicolumn{1}{c}{\multirow{2}{*}{}} & \multicolumn{1}{c}{\multirow{2}{*}{Method}} & \multicolumn{3}{c}{Shapelet-space Score} \\
\multicolumn{1}{c}{} & \multicolumn{1}{c}{} & All & \begin{tabular}[c]{@{}c@{}}Changing\\ Trend\end{tabular} & \begin{tabular}[c]{@{}c@{}}Continued\\ Trend\end{tabular} \\
\hline
\multirow{2}{*}{cases} & Shapelet Ensemble & 0.57 & 0.42 & 0.58 \\
 & COVIDhub Ensemble & 0.40 & -0.025 & 0.45 \\
 \hline
\multirow{2}{*}{deaths} & Shapelet Ensemble & 0.67 & 0.77 & 0.66 \\
 & COVIDhub Ensemble & 0.60 & 0.55 & 0.60 \\
 \hline
\end{tabular}
\end{table}

\subsection{Agreement vs Performance}
In real-time forecasting, it is advantageous to know when the forecasts are reliable. We hypothesize that the inter-model agreement is a good indicator of forecast reliability. To assess the relationship between model performance based on shapelet-space score and inter-model agreement, we plot them for each state and each week as shown in Figure~\ref{fig:perf_vs_ag}. We notice that the points corresponding to the Shapelet Ensemble are more tightly concentrated around the diagonal. This implies that when we see a high agreement, we are more likely to see a high performance. Similarly, a low agreement is likely to produce a low performance. The same pattern is observed in cases as well as deaths.
\begin{figure*}[!ht]
\centering
     \begin{subfigure}[b]{0.45\textwidth}
         \centering
         \includegraphics[width=\textwidth]{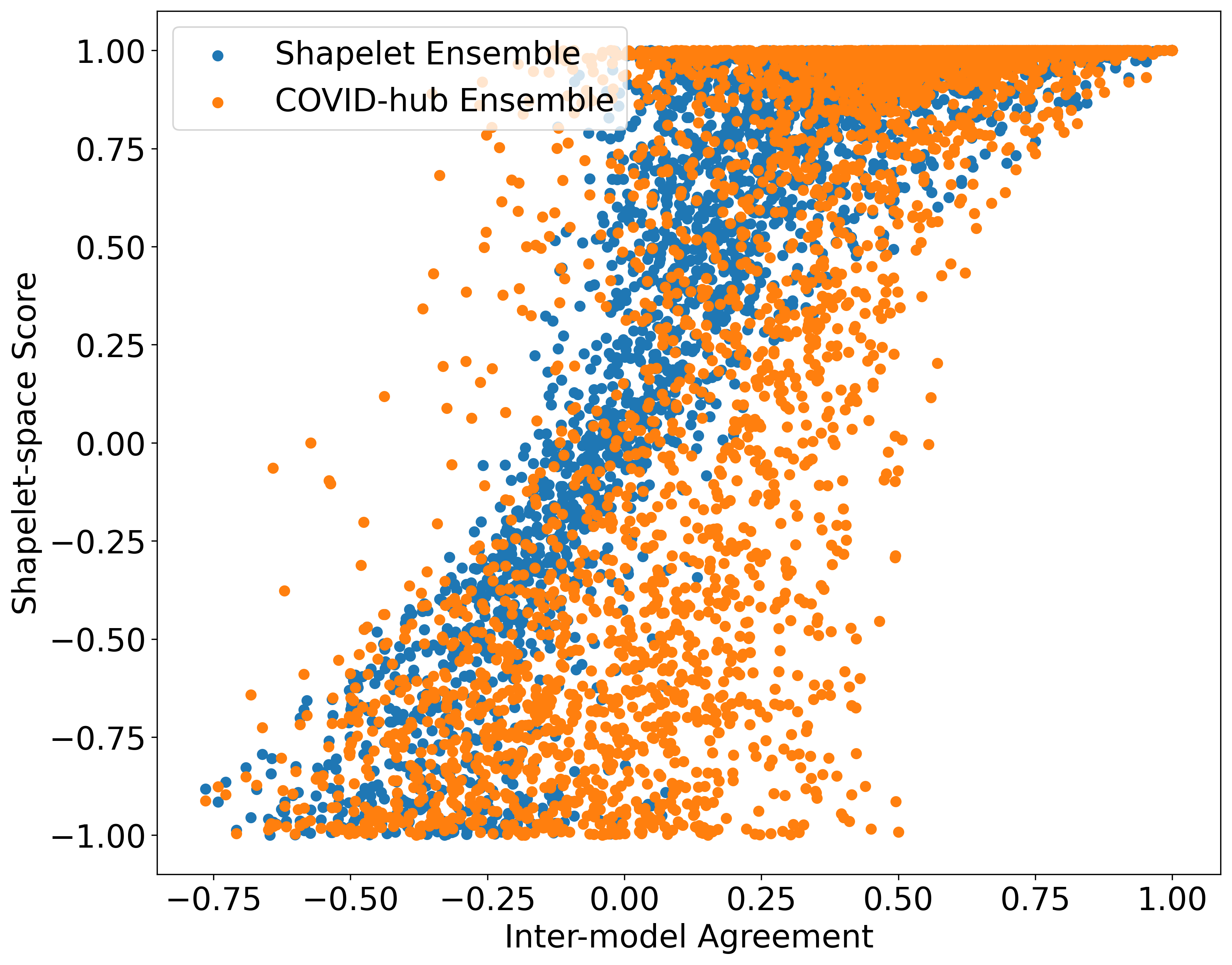}
         \caption{Cases}
         \label{fig:perf_vs_ag_cases}
     \end{subfigure}
     %\hfill
     \begin{subfigure}[b]{0.45\textwidth}
         \centering
         \includegraphics[width=\textwidth]{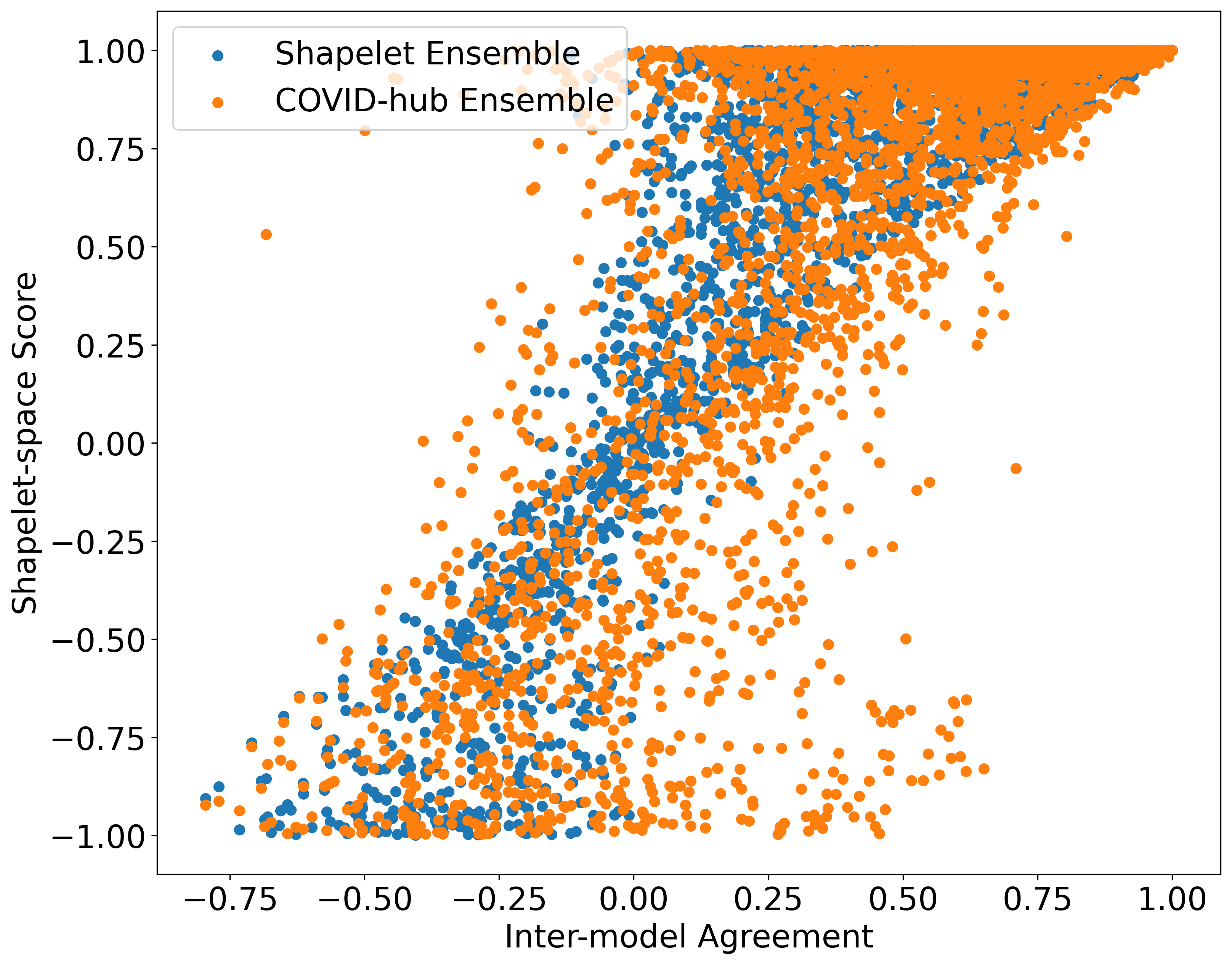}
         \caption{Deaths}
         \label{fig:perf_vs_ag_deaths}
     \end{subfigure}
        \caption{Shapelet-space Score vs Inter-model Agreement for the two models.}
        \label{fig:perf_vs_ag}
\end{figure*}
\begin{figure*}[!ht]
\centering
     \begin{subfigure}[b]{0.45\textwidth}
         \centering
         \includegraphics[width=\textwidth]{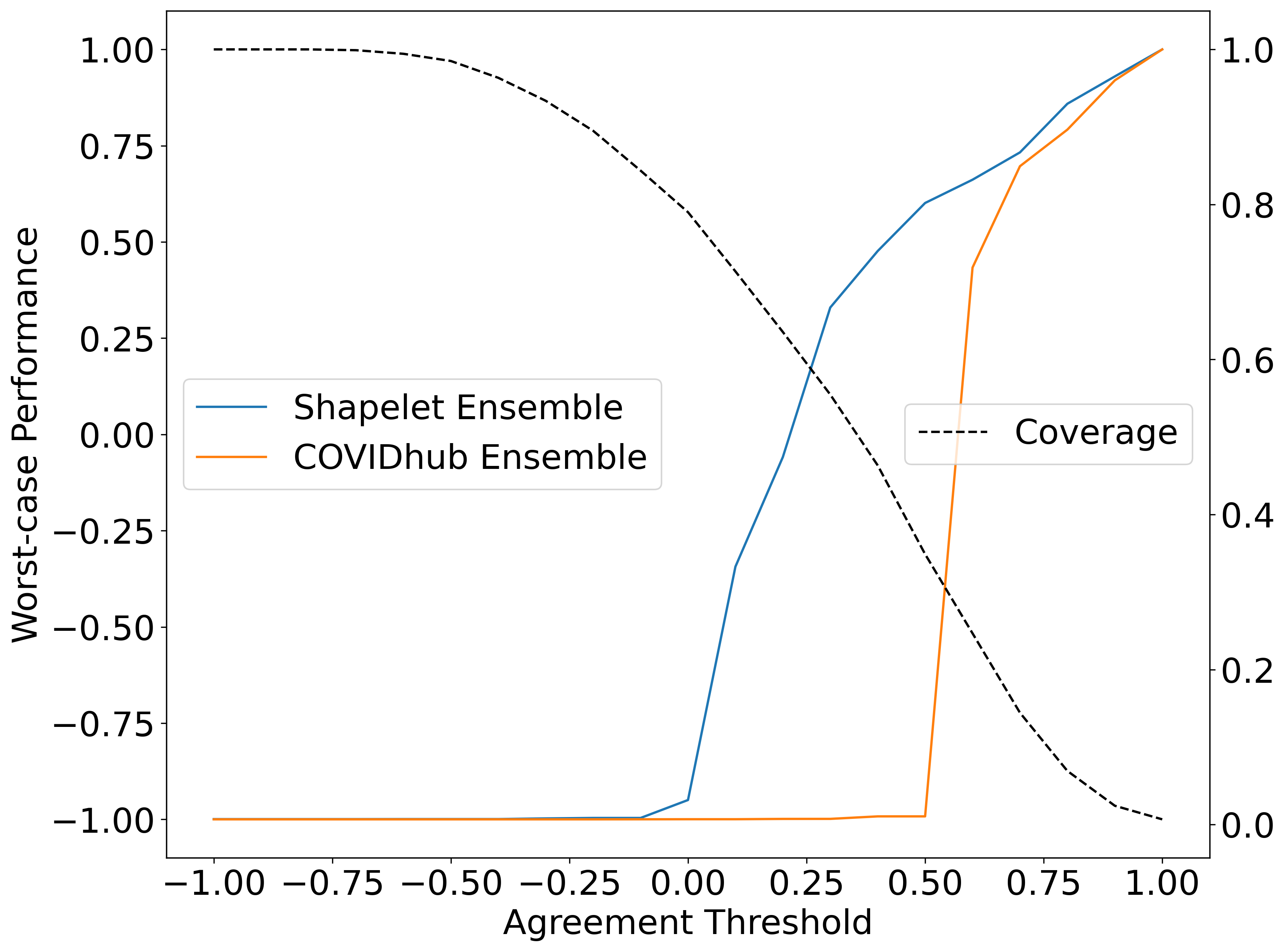}
         \caption{Cases}
         \label{fig:wperf_vs_ag_cases}
     \end{subfigure}
     %\hfill
     \begin{subfigure}[b]{0.45\textwidth}
         \centering
         \includegraphics[width=\textwidth]{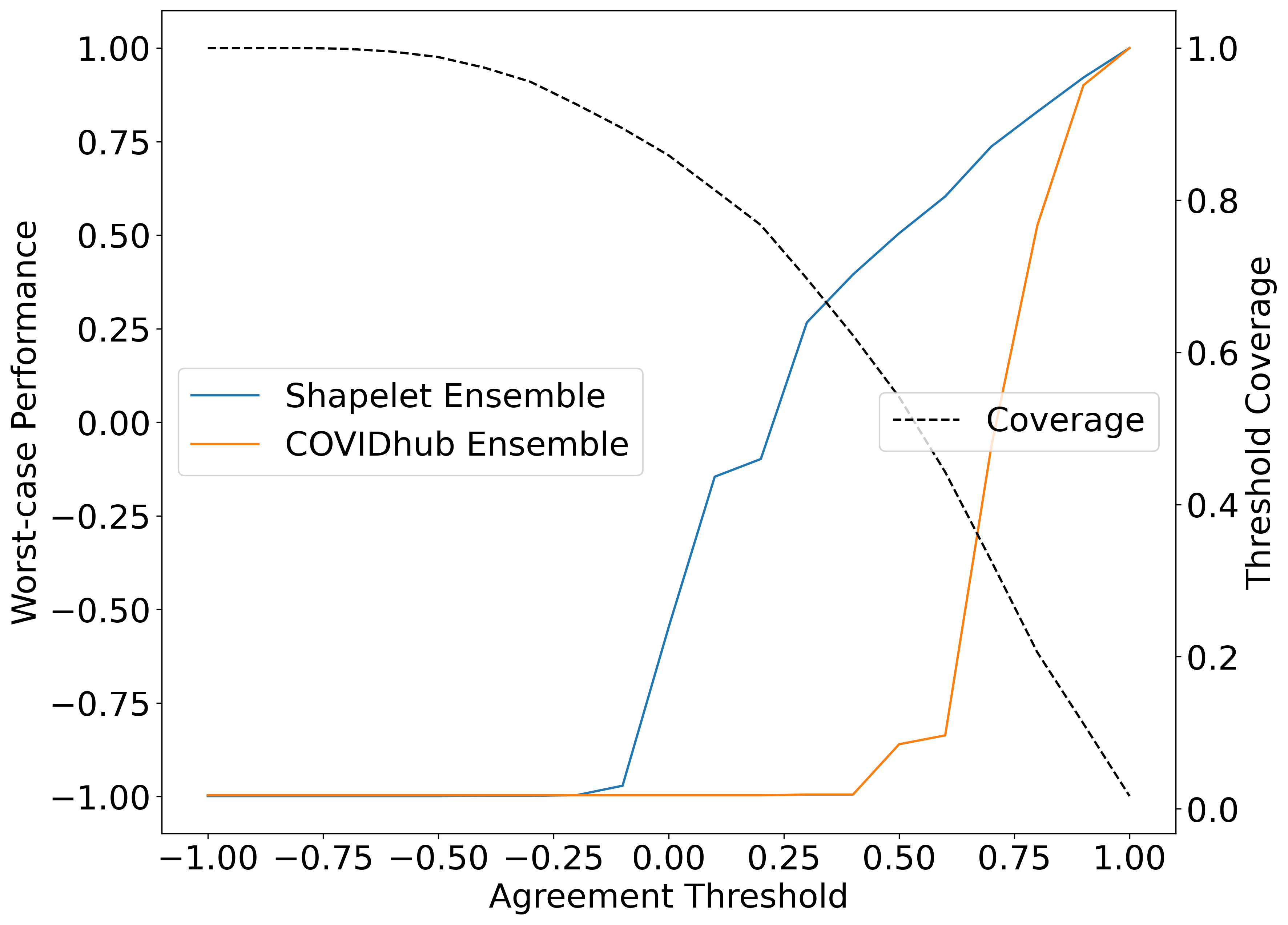}
         \caption{Deaths}
         \label{fig:wperf_vs_ag_deaths}
     \end{subfigure}
        \caption{Worst-case performance observed for a given threshold of Inter-model Agreement.}
        \label{fig:wperf_ag}
\end{figure*}

To further demonstrate our claim, we studied how the level of agreement affects performance in the worst cases. We fix a threshold for inter-model agreement. Among all the performance scores, we picked those obtained when the agreement was greater than the threshold. This gives us the worst-case performance observed for the given threshold of agreement. Figure~\ref{fig:wperf_ag} shows the results. We note that as agreement increases, the performance of both models improves. However, the worst-case performance of Shapelet Ensemble improves more rapidly. For instance, consider the inter-model agreement threshold of 0.25 for cases (Figure~\ref{fig:wperf_vs_ag_cases}). The fraction of times and states at which the agreement was above 0.25 is approximately 0.6, as indicated by the coverage curve. The worst performance at this threshold of agreement is close to 0.25 for Shapelet Ensemble and -1 for COVIDhub Ensemble.

\section{Conclusions and Future Work}
Infectious disease forecasting evaluation is traditionally done using measures that compare numerical predictions with ground truth. We proposed an evaluation based on shapes. To enable this evaluation, we define a set of shapes of interest called shapelets. We proposed a representation of a given shape in ``shapelet-space'' where each dimension represents the similarity of the given shape with one of the shapelets. We proved that two points are mapped close to each other in this representation if and only if the two shapes are similar. We demonstrated that the shapelet-space representation is able to capture the shapes of the trend well. Based on the representation we defined the shapelet-space score as the cosine similarity between the shapelet-space representation of the prediction and ground truth.

We transform the predictions of all the models submitted for COVID-19 cases and death forecasts on a given week into the shapelet-space representation. We took the average of the representation to obtain the shapelet-space ensemble. We demonstrated that this ensemble is able to capture the shape of future trends better than the COVIDhub ensemble model obtained by combining numerical forecasts of all the models. Particularly, for case forecasting during the time of changing trends, the shapelet ensemble is able to perform significantly better than COVIDhub ensemble. We also demonstrated that as the agreement among the models increases, the worst-case performance of the shapelet ensemble also increases. Thus, in prospective forecasting, agreement among the models can provide an early indicator of how reliable the forecasts are.  

We have only evaluated the shapes of the point forecasts. The real-time forecasting efforts also have probabilistic forecasts given in the form of quantiles. In future work, we will extend our methods to capture the expected shape generated by the probabilistic forecasts. We will also extend our analysis and evaluation to influenza forecasting.

\section*{Acknowledgement}

This work was supported by the Centers for Disease Control and Prevention and the National Science Foundation under the awards no. 2135784, and 2223933. Any opinions, findings, and conclusions or recommendations expressed in this material are those of the author and do not necessarily reflect the views of the National Science Foundation and the Center for Disease Control and
Prevention.

\bibliographystyle{IEEEtran}
\bibliography{epidemic,biblio}

% \vspace{12pt}
% \color{red}

\end{document}